\documentclass[sn-mathphys,Namedate]{sn-jnl}

\usepackage{multirow}%
\usepackage{mathrsfs}%
\usepackage[title]{appendix}%
\usepackage{xcolor}%
\usepackage{textcomp}%
\usepackage{manyfoot}%
\usepackage{booktabs}%
\usepackage{algorithm}%
\usepackage{algorithmicx}%
\usepackage{algpseudocode}%
\usepackage{listings}%

\usepackage{amsmath}        
\usepackage{amssymb}        
\usepackage{amsfonts}       
\usepackage{bbm}            
\usepackage{amsthm}         
\usepackage{caption}

\usepackage{graphicx}       
\usepackage{tikz}           

\usepackage{hyperref}

\usepackage[many]{tcolorbox}

\usepackage{empheq}

\definecolor{main}{HTML}{5989cf}    
\definecolor{sub}{HTML}{cde4ff}     
\tcbset{
    rounded corners,
    colback = white,
    before skip = 0.2cm,    
    after skip = 0.5cm      
}  
\newtcbox{\mcnerney}[1][]{%
    nobeforeafter, math upper, tcbox raise base,
    enhanced, colframe=blue!30!,
    colback=white, boxrule=1pt,
    #1}
\newtcbox{\new}[1][]{%
    nobeforeafter, math upper, tcbox raise base,
    enhanced, colframe=red!30!,
    colback=white, boxrule=1pt,
    #1}
\tcbset{highlight math style={boxsep=5mm, colframe=blue!30!, colback=white}}
\newtcolorbox{boxB}{
    boxrule = 1.5pt,
    colframe = main,
    rounded corners,
    arc = 5pt   
}

\newcommand\rp{\mathbb{R}_{>0}}

\newcommand\leigen{\kappa^{\text{max}}}
\newcommand\ind{\mathcal{I}}    
\newcommand\pref{\kappa}        
\newcommand\exoindus{\Psi}      
\newcommand\g{\hat{gdp}}        
\newcommand\defn[1]{\textbf{\textit{#1}}}
\newcommand\pdv[2]{\frac{\partial#1}{\partial#2}}
\newcommand\pdvtwice[2]{\frac{\partial^2 #1}{\partial#2^2}}
\newcommand\dv[2]{\frac{d #1}{d #2}}
\newcommand\domar{\theta}
\newcommand\OM{\mathcal{L}}

\newcommand\pchange{\hat{p}}    
\newcommand\rpchange{\hat{r}}   
\newcommand\fc{S_{tot}}
\newcommand\ls{\tilde{\ell}}
\newcommand\s{a_{ij}}
\newcommand\indprod{Y}
\newcommand\expeccond[2]{\mathbb{E} \left[ #1 | #2 \right]}

\newtheorem{theorem}{Theorem}
\newtheorem{corollary}{Corollary}
\newtheorem{prop}{Proposition}
\newtheorem{lemma}{Lemma}
\theoremstyle{definition}
\newtheorem{definition}{Definition}

\begin{document}

\title{Interdependent Total Factor Productivity in an Input-Output model}


\author[1,2]{\fnm{Thomas M.} \sur{Bombarde}}\email{thomas.bombarde@st-hughs.ox.ac.uk}

\author*[1]{\fnm{Andrew L.} \sur{Krause}}\email{andrew.krause@durham.ac.uk}

\affil*[1]{\orgdiv{Department of Mathematical Sciences}, \orgname{Durham University}, \orgaddress{\street{Upper Mountjoy Campus, Stockton Road}, \city{Durham}, \postcode{DH1 3LE}, \country{United Kingdom}}}

\affil[2]{\orgdiv{St Hugh's College}, \orgname{University of Oxford}, \orgaddress{\street{St Margaret's Rd}, \city{Oxford}, \postcode{OX2 6LE}, \country{United Kingdom}}}


\abstract{ABSTRACT TEXT HERE}

\abstract{ Industries learn productivity improvements from their suppliers. The observed empirical importance of these interactions, often omitted by input-output models, mandates larger attention. This article embeds interdependent total factor productivity (TFP) growth into a general non-parametric input-output model. TFP growth is assumed to be Cobb-Douglas in TFP-stocks of adjacent sectors, where elasticities are the input-output coefficients. Studying how the steady state of the system reacts to changes in research effort bears insight for policy and the input-output literature. First, industries higher in the supply chain see a greater multiplication of their productivity gains. Second, the presence of `laggard' industries can bottleneck the the rest of the economy. By deriving these insights formally, we review a canonical method for aggregating TFP -- Hulten's Theorem -- and show the potential importance of backward linkages.}

%
%
%




\maketitle

\section{Introduction}

 In the long-run, economies grow by increasing total factor productivity (TFP). By changing the nature of industries and the interrelations between them, economies tap into new opportunities for growth. We motivate this paper by discussing the questions in economic development that input-output (IO) models help answer, the recent renew of interest in IO economics, and the observations of TFP-growth interacting with IO linkages. 

Early development economics argued for changes in the nature of industries as a precondition for take-off \citep{Monga2019StructuralDestiny}. The challenge is to qualify this nature. A neat distinction between agriculture, manufacturing, and services as in early models \citep{Kuznets1955} is less pertinent to modern economies, where high-skill service industries may play leading sector roles previously attributed to manufacturing, and informal manufacturing may serve as reserves of surplus-labour that were previously attributed to agriculture \citep{Lamba2020DynamismDevelopment, Mensah2022StructuralAfrica}. At the same time, premature deindustrialisation of developing economies calls for a finer understanding of the interaction between industrial specialisations and economic growth \citep{Rodrik2016, HausmannRodrikEconDevasSelfDiscovery, Lin2013LessonsTransformation}. Beyond what industries \textit{do}, defining how important an industry is to the economy at a state in time should consider its place in the supply chain. 

Input-Output (IO) economics provides quantitative measures on an industry's economic importance.  It assumes that an economy can be divided into industries characterised by distinct output quantities. `Technical coefficients' $a_{ij}$, the quantity of input from the $j$th industry required per unit of output of the $i$th industry, are stored in an adjacency matrix that describes a network of industries \citep{Raa2006TheInefficiency}. The output-multiplier is a centrality metric on this network for the position of an industry in the supply chain \citep{Mcnerney2021HowGrowth}.  Hulten's Theorem \citep{Hulten1978GrowthInputs} abstracts from these dependencies to aggregate TFP growth rate as the sum of industry-level TFP growth rates weighted by the their `Domar weights', the ratio of industry sales to final consumption. If $\g$ is the growth rate of final consumption (henceforth `economic growth'), $\domar_i$ the Domar weight of an industry $i=1,...,n$, and $\gamma_i$ its growth rate of TFP, then we formulate Hulten's result as
\begin{equation}\label{Hultens_Result}
    \pdv{\g}{\gamma_i} = \domar_i.
\end{equation}
In this view, the Domar weights of industries are sufficient to understanding the impact of industry-level TFP-shocks on economic growth. In short, eq. \eqref{Hultens_Result} suggests that IO dependencies are irrelevant to development policy. 

Recent literature explores the role of IO dependencies in propagating economic shocks. \cite{Gabaix2009} and \cite{Acemoglu2012TheFluctuations} use Hulten's Theorem to show how a fat-tailed distribution of Domar weights exposes the economy to shocks cascading from central industries. However, as  commented by \cite{Baqaee2019}, Hulten's Theorem is strikingly unintuitive, as it suggests that deleting the supermarket chain Walmart from the US economy would have the same effect as deleting all electricity providers, their Domar weights being roughly equal. Most challenges to Hulten's Theorem focus on second-order effects the theorem omits. \cite{Acemoglu2017} model a fat-tailed distribution of shocks, second order effects, and Domar weights that trigger economic collapse. \cite{Baqaee2019} assess susceptibility to economic disaster through second-order effects, but do so even for smaller shocks that are symmetric across industries by using a range of production functions. In contrast, this article breaks Hulten's Theorem in the first-order by considering interdependent TFP, rather than introducing mark-ups, or industry goods as substitutes to labour as in \cite{jackson2019automation}. While the previous studies renew interest in IO economics, they concentrate on supply- and demand- between industries,  and thus differ from this report's interest on the interaction between industrial structure, TFP, and economic growth.

\cite{Mcnerney2021HowGrowth} formalise the relationship between IO structure and economic growth. Assuming a non-parametric production function with constant returns to scale, the authors show how the length of its supply chains accelerates the rate of price declines. These price declines drive economic growth. A note of interest is the correlation between a key measure of supply chain length -- the output multiplier -- and productivity growth, suggesting that some level of interdependence between both variables may be at play. Similarly \cite{Shuoshuo} finds a significant relationship between centrality measures for industries and real GDP growth for the US economy from 1970-2017. This relationship holds controlling for Domar weights, which suggests a relationship between centrality and economic growth beyond Hulten's Theorem. Building on \cite{Baqaee2019}, \cite{Shuoshuo} considers second order effects of productivity shocks to show that two economic networks with the same Domar weights can witness different effects on GDP of industry-level TFP shocks. However, \cite{Shuoshuo} models TFP as a random walk and in doing so maintains productivity growth to be independent of IO linkages. Instead, \cite{Mcnerney2021HowGrowth}'s correlation between industry-level TFP and Katz centrality suggests that the IO network may directly affect TFP growth. We consider how such a relationship modifies eq. \eqref{Hultens_Result} in a first order approximation.

The role IO networks play in productivity growth is a growing area of IO economics. An early exception is \cite{Wolff2000EnginesRaa}'s investigation of the subtle productivity-enhancements of the US 1990s Information and Communications Technology (ICT) revolution. Specifying a model of interdependent productivity improvements, the authors establish a multiplier effect in the returns to R\&D of industries. They impute ICT as the leading productivity growth sector of the US in the 1980s. However, because they specify a model of TFP growth rates directly without a link between productivity growth of GDP, the authors shy away from assessing how productivity shocks affect economic growth. \cite{acemoglu2023bottlenecks} provide a competing explanation for the US productivity slowdown by considering TFP improvements to dependent on IO links, stressing the role of `laggard' sectors. We follow this line of thought, but contextualise model interdependent TFP within Hulten's Theorem and the relationship to the output multiplier.

We use \cite{Mcnerney2021HowGrowth}'s IO model to relate productivity growth rates to economic growth through price changes. We also take labour as only factor-input, which might underestimate interactions such as learning by doing from capital accumulation, but helps focus on intermediate inputs. To endogenise productivity, we differ from \cite{Wolff2000EnginesRaa} and instead emulate \cite{pichler2018technological}'s discussion of knowledge creation, in the sense that we first describe growth of stocks, then derive rates. This framework makes our arguments less empirically tractable than \cite{Wolff2000EnginesRaa}, but allows us to theoretically link productivity growth to the output multiplier. We find a positive relationship to be conditional on  a high level of interdependence in input. Moreover, we show the relationship in eq. \eqref{Hultens_Result} to be a special case of our model. Our caveat to Hulten's aggregation through Domar weights comes from potentially `laggard' sectors as in \cite{acemoglu2023bottlenecks}, though we focus on how this relates to the output multiplier, and emphasise the importance of backward linkages.

 Our model stresses that the industries of highest return to investment in R\&D for economic growth depend on both existing IO relations and intra-sector returns to investment. We see this as formalising two policy insights. First, industries higher in the supply chain may be preferred over those that have higher intra-industry returns to R\&D, but fewer downstream buyers. Second, the presence of `laggard' industries within the economy may have important detrimental effects on the rest of the economic environment, which cannot be considered in isolation. Overall, this article lays out a theoretical mechanism for the importance of TFP interdependencies. It articulates this with a toy model, but does not lend itself to empirical estimation.
 
 The rest of the paper is organized as follows.  Section \ref{S:StandardIO} describes the IO model as in \cite{Mcnerney2021HowGrowth}. In Section \ref{S:interdepTFP} we develop our model of interdependent TFP growth. Section \ref{S:HultenExperiment} investigates the relationship between productivity improvements and output multiplier. Finally, Section \ref{Cobb-Douglas} studies a Cobb-Douglas production function to consider how our model's response to TFP shocks differs qualitatively from \cite{Hulten1978GrowthInputs}. Section \ref{S:Discussion} concludes with a discussion of the implications and limitations of our work.

 Unless specified otherwise, all scalars are in $\rp$, and $i, j \in \{1,...,n\}$. We use $\dot{ }$ to denote a time ($t$) derivative, and $\hat{ }$ to denote a  growth rate, e.g. $\hat{h} = \dot{h}/h = \pdv{h}{t} / h$. When placed over a vector, these refer to component-wise derivatives. Proofs are left to appendices.

\section{The Standard Input-Output Model}
\label{S:StandardIO}

In this section, we set the core definitions of an unparameterised IO model as in \cite{Mcnerney2021HowGrowth}. We first state the assumptions of our IO model, which are also summarised in Appendix \ref{Assumptins_Appendix}.  We then derive a corollary of Hulten's Theorem given in eq. \eqref{Hultens_Result}. This corollary describes a direct relationship between the response of economic growth to industry-level TFP growth and the industry's ratio of sales to final consumption. 

Consider $N$ industries, each defined by a production function $F_i$, for $i=1,...,N$. 
The production function of each industry is a map from inputs to a quantity of its unique good. We formalise this by $F_i(X_{1i},...,X_{Ni}, L_i, Z_i)$ where $(X_{1i},...,X_{Ni}) \equiv \boldsymbol{X}_i$ denotes all intermediate inputs from industries 1 through N to industry $i$, and $L_i$ is the units of labour consumed by industry $i$ from a fixed stock $L = \sum_{j=1}^N L_j$.  The term $Z_i$ scales these inputs by the industry's TFP, and is expected to increase in time with a growth rate $\gamma_i \equiv \hat{Z}_i$. We assume diminishing marginal returns (\textbf{A1}), expressed as $\pdv{F_i}{k} > 0, \text{ }\pdvtwice{F_i}{k} < 0$ for each ${k} \in \{X_{1i},...,X_{ni}, L_i\}$. We follow this with the stronger condition that $F_i$ has constant returns to scale to physical inputs (\textbf{A2}), so that doubling all inputs doubles output. That is, for any integer $b$,
\begin{equation}
F_i(bX_{i1},...,b L_i,Z_i) = b F_i(X_{i1},...,L_i,Z_i).
\end{equation}
The interest in these assumptions lie in the result that the elasticities
of $F_i$ with respect to all inputs $k$ then sum to $1$ \citep{Carvalho2019}, which is crucial for the following analysis of intermediate input flows. That is,
\begin{equation}\label{F_Elasticity_Result}
    \sum_{k \in {X_{1i}, ..., X_{1N}, L_i}}\pdv{\ln F_i}{\ln k} = 1.
\end{equation}
Whilst the production function describes output given a set of inputs, observed output $Y_i$ is determined by the input quantities that maximise the industry's total revenue -- its output times its good's unit price $p_i$ -- subject to the industry's budget constraint. That is, we assume $Y_i$ is the value function
\begin{align}
 Y_i \equiv &\text{max}_{\boldsymbol{X}_{i}, L_i} {F_i p_i}, \\\text{ subject to: } L_i w_i + \sum_{j=1}^N p_j X_{ji} &\leq C_i p_i + \sum_{j=1}^N X_{ij}p_i, \label{e:COM_industries}
\end{align}
where $w$ is the wage of labour, $p_j$ is the price of the $j$th industry's good, and $C_i$ is the consumption by households of good $i$ (\textbf{A4}). The budget constraint states that expenditure of industry $i$ cannot exceed its income from household consumption and its sales as intermediate input. We assume the representative household faces a similar constraint (\textbf{A3}), namely that it maximises a convex continuous utility function with consumption as its only argument subject to 
\begin{equation}
    \sum_{i=1}^n C_i p_i \leq \sum_{i=1}^n L_i w.\label{e:COM_households}
\end{equation}
This last condition assures that household expenditure on goods does not exceed household income from wages. The optimisation schedule for both industries and households makes eq. \eqref{e:COM_industries} and \eqref{e:COM_households} equalities. As industry level TFP $Z_i$ increases, maximisation of observed outputs $Y_i$ reduces industry price, which we refer to as $\pchange_i$, or in real terms as $\rpchange_i \equiv \pchange_i/w$. The rest of this section is concerned with the dependency of industry prices on input prices.

To discuss intermediate inputs, we label monetary flows from industry $i$ to $j$ as $m_{ij} \equiv {X}_{ji} p_j$. Normalising $m_{ij}$ by industry's $i$'s total expenditure gives its share of spending on input $j$: ${a}_{ij} \equiv {m_{ij}}/\left(\sum_{j=1}^N m_{ij} + L_iw\right)$. These shares are stored in the $N \times N$ adjacency matrix of the IO network, $\boldsymbol{A}$. Under assumptions \textbf{A1}-\textbf{A4}, it follows that $a_{ij} = \pdv{\ln(Y_i)}{\ln(X_{ji})}$ \citep{Mcnerney2021HowGrowth, Carvalho2019}. That is, $a_{ij}$ equals the elasticity
of output with respect to intermediate input $j$, a result referred to as allocative efficiency. Since $F_i$ has constant returns to scale, it can be shown that $\sum_{k \in {X_{1i}, ..., X_{1N}, L_i}}\pdv{\ln F_i}{\ln k} = 1$. Assuming $\pdv{\ln F_i}{\ln L_i} > 0$, it follows from \eqref{F_Elasticity_Result} that the sum of all intermediate input elasticities is less than one. As a result, the matrix $\boldsymbol{A}$ qualifies as substochastic, in that its rows elements -- equal to the constant intermediate input elasticities -- are positive and sum to less than 1 \citep{Carvalho2019}. This is useful because it guarantees the existence of $\boldsymbol{H} \equiv (\boldsymbol{I} - \boldsymbol{A})^{-1}$, where $\boldsymbol{I}$ is the identity matrix and we refer to $\boldsymbol{H}$ as the Leontief inverse. The following result is presented in \citep{Mcnerney2021HowGrowth}. It  makes the Leontief inverse central to the IO model. A proof is provided in Appendix \ref{Standard_IO_Appendix}.

\begin{theorem}
In the IO economy under assumptions \textbf{A1}-\textbf{A4}, the vector of real price changes is\begin{equation}
    \boldsymbol{\hat{r}} = - \boldsymbol{H} \boldsymbol{\gamma}
\end{equation}. \label{t:pricechanges}
\end{theorem}

We can provide intuition to Theorem \ref{t:pricechanges} by noting that the coefficients $H_{ij}$ are the sum all possible weighted paths between $i$ and $j$, since $\boldsymbol{H} = (\boldsymbol{I} - \boldsymbol{A})^{-1} = \sum_{n-1}^{\infty} \boldsymbol{A}^n$. Thus, price declines are amplified as they propagate through the supply chain. We expect, therefore, that the more important an industry to the economy's supply chain, the larger the price changes induced by its productivity improvements. It follows that, all else equal, economies with longer supply chains grow faster. To assess this claim, we define the output multiplier $\boldsymbol{\OM} = \boldsymbol{H} \cdot \boldsymbol{1},$ with $\boldsymbol{1}$ a vector of 1's. By summing over rows of the Leontief inverse, a higher output-multiplier $\OM_i$ reflects a greater importance of industry $i$ in the economy's supply chain. \cite{Mcnerney2021HowGrowth} provide a discussion of the economic, probabilistic, and ecological intuitions behind the output-multiplier. 

Economic growth is here considered in terms of total final consumption, which denotes industry sales to households as $\fc \equiv \sum_{i=1}^n C_i p_i$. Denoting household spending shares $\tilde{c}_i \equiv C_i p_i / \fc$ allows us to define the nominal inflation index as $\pchange = \boldsymbol{\pchange} \cdot \tilde{c}$, and the weighted average output multiplier of the economy  $\bar{\OM} = \boldsymbol{\OM} \cdot \boldsymbol{\tilde{c}}$. By weighting for consumption shares, $\bar{\OM}$ captures how important supply chains are to household consumption. Note that $\fc$ is distinct from the total sales of industries,  $m_i \equiv \sum_{j=1}^n m_{ij} + C_i p_i$, which sum to total monetary output of the economy, $\mathcal{M} = \sum_{i=1}^n m_i.$ Both measures feed into a weighted average of the economy's TFP-growth rates as $\bar{\gamma} = \boldsymbol{\gamma}\cdot\boldsymbol{{m}}/\mathcal{M}$, which captures the importance of TFP improvements in the overall economy. We consider economic growth as the increase of household consumption in real terms, $\g \equiv \fc - \pchange$, which relates to the previous descriptive statistics in the next result.

\begin{theorem}
Assuming \textbf{A1}-\textbf{A4}, economic growth $\g$  is the product of average productivity improvements $\tilde{\gamma}$ weighted by industry sales and the average output multiplier $\bar{\mathcal{L}}$ weighted by household consumption shares, 
\begin{equation}
    \g = \tilde{\gamma} \bar{\mathcal{L}}.
\end{equation}
\label{t:OMg}
\end{theorem}

Theorem \ref{t:OMg} is discussed at greater length by \cite{Mcnerney2021HowGrowth}. There are several insights to note. First, Theorem \ref{t:OMg} captures the importance of consumer-oriented supply chains: $\bar{\OM}$ increases with the length of the supply chain reflected in the output multiplier, but scales by each industry's importance to final demand. All else equal, increasing importance of industry $i$ as an intermediate input to others should increase economic growth, however, less so for industries that represent a smaller share of final demand ($\tilde{c}_i$). Second, the model stresses productivity improvements as the driving force of economic growth. This is because $\bar{\OM}$ is empirically slow-moving, and does not fluctuate significantly over the 17-years considered by \cite{Mcnerney2021HowGrowth}. Productivity improvements, on the other hand, display greater year-to-year volatility. In our stylised model, productivity improvements are the short term driver of economic growth. Third, eq. \eqref{Hultens_Result} appears as a corollary. Denote the ratio of an industry's total sales to final consumption as its `Domar weight' $\domar_i \equiv m_i/\fc.$

\begin{corollary}
The response of economic growth to changes in the growth rate of TFP is given by
\begin{equation}
\pdv{\g}{\gamma_i} = \domar_i.
\end{equation}
\label{t:hultendynamic}
\end{corollary}
See Appendix \ref{Standard_IO_Appendix} for a proof.

Hulten's Theorem originally expressed the aggregate TFP level, computed by the residual of total gross domestic product and economic inputs, as the sum of industry level TFP weighted by their Domar weights. Corollary \ref{t:breakhulten} emphasises this relationship by stating the effect of industry TFP shocks to be proportional to industry Domar weights. Whilst eq. \eqref{Hultens_Result} is a welcome simplification reducing productivity accounting to a weighted sum of industry statistics, it also rests on important assumptions. We have so far assumed independent productivity improvements across industries. \cite{Mcnerney2021HowGrowth} notes a correlation between an industry's TFP growth rates and its output multiplier. Interdependence between both variables threatens to cloud policy implications for Theorem \ref{t:OMg}, since plans to alter $\bar{\gamma}$ would have to carefully consider any subsequent effects on $\bar{\OM}$. In the rest of this paper, we review the assumption of TFP independence with respect to IO linkages, and explore how this changes the first order effect of industry-level TFP shocks. 

\section{Interdependent TFP-growth on the industry network}\label{S:interdepTFP}
This section presents a model of interdependent TFP-growth and derives the equilibrium in terms of growth rates. By TFP-growth, we mean the time derivative of $Z_i$. We assume a Cobb-Douglas knowledge production function $F_i(X_{1i},...,X_{Ni},L_i,Z_i)$ (\textbf{A5}),
\begin{equation}
\dot{Z}_i  \equiv  ({E_i}e^{\lambda_i t})^{\alpha}\exoindus_i  \prod_{j=1}^N Z_{j} ^{\beta a_{ij}}, \label{d:tfp}
\end{equation}
where $a_{ij}$ is the coefficient of the IO network's adjacency matrix $\boldsymbol{A}$, and $\alpha$, $\beta$ are elasticities in $(0,1)$. TFP-growth in eq. \eqref{d:tfp} depends explicitly on three terms. First, each industry has a unique and innate TFP growth factor $\exoindus_i$, which captures differential rates of technological progress in specific sectors \citep{FinkReeves2019}, and is assumed constant with respect to time as an initial simplification. Second, the term $E_{i}$ is a measure of the sector's initial endowments that might be associated to initial capital ($K_{i,0}$) or R\&D ($R_{i,0}$) as in the endogenous growth literature \citep{Romer1994, ErturKoch2007, pichler2018technological}. This grows at an exponential rate $\lambda_i$, and captures `learning by doing' \citep{Arrow1962TheDoing}. Third, there are TFP spillovers over the IO network, such that the elasticity of $Z_i$ with respect to  $Z_j$ is $\beta a_{ij}$ for some constant $\beta<1$ in $\rp$. We expect this to be the case because for countries at the technological frontier, innovations arise within a specific sector before adoption in adjacent sectors. We thus expect TFP to increase in the sector of innovation first, before diffusing to other sectors as a function of inter-sectoral interaction. For developing countries below the frontier, we expect a similar process because implementing foreign technologies requires adaptation to a novel environment which leading sectors succeed in first, before a similar diffusive process to neighbouring industries. The use of IO flows as a measure of industry interactions is intuitive, but also finds corroboration in preceding TFP-spillovers on the World Input-Output Database \citep{Serrano2017}, 
within US industries \citep{Raa2006TheInefficiency, Wolff2011} and in EU industries \citep{ProdGrowthEurope}.

We next show that the steady state of TFP growth rates ($\boldsymbol{\lambda}$) determined by eq. \eqref{d:tfp} is the globally stable solution $\boldsymbol{\gamma}_0$ to a system of generalised Lotka-Volterra equations.  These systems take the form $\boldsymbol{\dot{g}} = \boldsymbol{D(g)}\left(\boldsymbol{\rho} + \boldsymbol{\tilde{A}}\boldsymbol{g}\right)$, for $\boldsymbol{g}$ a vector of populations, $\boldsymbol{D}(\boldsymbol{g})$ the diagonal matrix with $\boldsymbol{g}$ on the diagonal, $\boldsymbol{\rho}$ the reproduction rate of the populations, and $\boldsymbol{\tilde{A}}$ the interaction matrix of the Lotka-Volterra system. These systems are commonly used to describe predator-prey and competitive dynamics. Given that the interaction matrix is invertible, the steady state of such systems (when $\boldsymbol{\dot{g}}=\boldsymbol{0}$) is given by $\boldsymbol{g}=-\boldsymbol{\tilde{A}}^{-1} \boldsymbol{\rho}$ \citep{hofbauer1998evolutionary}. 

We  first note that $\dot{\gamma_i}  = {\ddot{Z_i}}/{{Z}_i} - \gamma_i $. Then, taking the logarithm of eq. \eqref{d:tfp},
and differentiating with respect to time,
 \begin{align}
\frac{\ddot{Z_i} }{\dot{Z_i} } &=  \alpha \lambda_i + \sum_{j=1}^N \beta a_{ij} \hat{Z_j }.
 \end{align}
Multiplying both sides by ${\gamma}_i ={\dot{Z}_i }/{{Z}_i }$ and subtracting ${\gamma}_i ^2$,
\begin{align}
\frac{\dot{Z}_i }{{Z}_i } \frac{\ddot{Z_i} }{\dot{Z_i} }  - {\gamma}_i ^2 &= \frac{\dot{Z}_i }{{Z}_i }\left(\alpha \lambda_i + \sum_{j=1}^N \beta \hat{Z_j } a_{ij}\right) - {\gamma}_i ^2,
\end{align}
which rearranges to
\begin{align}
 \dot{\gamma_i}  &= {\gamma_i}  \left(\alpha \lambda_i + \sum_{j=1}^N \beta \gamma_j a_{ij}\right) - {\gamma}_i ^2.
 \end{align}
When writing this system in matrix form,
 \begin{align}
 \dot{\boldsymbol{\gamma}} =& \textbf{D}\left(\boldsymbol{\gamma} \right)\left(\alpha \boldsymbol{\lambda} - \left(\boldsymbol{I} - \beta \textbf{A} \right)\boldsymbol{\gamma} \right), \label{e:gammaLV}
 \end{align}
 where $\textbf{D}\left(\boldsymbol{\gamma} \right)$ is the diagonal matrix with diagonal entries given by $\boldsymbol{\gamma} $, we find a competitive Lotka-Volterra system  \citep{hofbauer1998evolutionary,baigent2016lotka} with interaction matrix $-\left(\boldsymbol{I} - \beta \textbf{A} \right)$. Because $\boldsymbol{A}$ is substochastic, its largest eigenvalue is less than 1 \citep{Carvalho2019}. Exploiting the fact that $\beta < 1$ and applying Lemma \ref{L:spectralraiduslinear} (Appendix \ref{Appendix_LV_Invertibility}), the largest eigenvalue of $- \left(\boldsymbol{I} - \beta \textbf{A} \right)$ must thus be strictly negative. This is useful because Theorem 2 in \cite{Plemmons1977M-matrixM-matrices} then guarantees that $- \left(\boldsymbol{I} - \beta \textbf{A} \right)$ is an invertible M-matrix. The steady state value is thus
\begin{align}\label{LV_ss}
    \boldsymbol{\gamma}_{0} = \alpha  \left(\boldsymbol{I} - \beta \textbf{A} \right)^{-1} \boldsymbol{\lambda} = \alpha \boldsymbol{H}_{\beta} \boldsymbol{\lambda},
\end{align}
which is globally stable, e.g. any vector of initially positive growth rates $\boldsymbol{\gamma}$ converges to $\boldsymbol{\gamma}_0$ (this arises as any M-matrix is Lyapunov diagonally stable \citep{sun2023gallery}; see \cite{hofbauer1998evolutionary} for further details or Appendix \ref{Appendix_LV_Invertibility} for an alternative proof). To understand some of the implications of this steady state, we can write the component form as 
\begin{equation}
\gamma_{0,i} = \alpha \lambda_i/(1-\beta a_{ii}) + \beta \sum_{i \neq j}^N \gamma_{0,j} a_{ij}/(1-\beta a_{ii}),
\end{equation}
where the first term is a network-independent term as studied in the endogenous growth literature \citep{pichler2018technological}. The second term is the weighted interaction from all neighbouring growth rates. Note that these interdependencies may result in lower growth rates for industries with relatively higher growth rate parameters
TFP growth given by eq. \eqref{d:tfp} determines a steady state system that reflects both industry-level advantages and the role that the wider economic environment plays in further productivity gains. 

We now turn to relating this form for $\gamma_0$ to the IO model presented in Section \ref{S:StandardIO} and the correlation empirically observed between output multipliers and productivity growth rates \citep{Mcnerney2021HowGrowth, Shuoshuo}.

\begin{prop}
For the IO model described in Section \ref{S:StandardIO} with assumptions \textbf{A1}-\textbf{A4} and the Cobb-Douglas model of TFP (\textbf{A5}), we expect the output multiplier to be positively correlated with TFP growth rates when 
\begin{align}
\bar{\lambda} > \frac{\left(1 - \beta\right)}{\alpha} \bar{\gamma},
\end{align}
where  the arithmetic means of TFP growth rates are given by $\bar{\gamma}=\sum_i \gamma_i/n$ and returns to endowments are given by $\bar{\lambda} =\sum_i \lambda_i/n$.
\label{p:corrOMgamma}
\end{prop}

By introducing the specific form of interdependent TFP growth in eq. \eqref{d:tfp}, we obtain a natural link between steady state TFP growth rates and output multipliers. In the following, we consider industry $\gamma_i$ and $\lambda_i$ to follow a uniform distribution, an argument taken from \cite{Mcnerney2021HowGrowth}. We write TFP growth rates
\begin{equation}
    \gamma_i = \bar{\gamma} + \Delta\gamma_i, \label{e:gammadecomposed}
\end{equation}
the sum of the average TFP growth rate across industries $\bar{\gamma}$ and an industry deviation $\Delta \gamma_i \in [-D,D]$ for some fixed $D>0$. The correlation between the elements of the Leontief inverse $H_{ij}$ with $\gamma_i$ are low, so we assume $\sum_{j=1}^n \mathbb{E} \left[ \Delta\gamma_j H_{ij} | \OM_i \right] = 0$. Similarly, we start by assuming the growth rate of initial endowments ($\lambda_i$) to be independent of the Leontief inverse, so that $\lambda_i = \bar{\lambda} + \Delta \lambda_i$ and $\sum_{j=1}^n \mathbb{E} \left[ \Delta \lambda_j H_{ij} | \OM_i \right] = 0$.
Rearranging the steady state given by eq. \eqref{LV_ss} we find, 
\begin{align}
    \boldsymbol{H}_{\beta}^{-1} \boldsymbol{\gamma}_0 &= \alpha \boldsymbol{\lambda}.
\end{align}
Expanding $\boldsymbol{H}_{\beta} = \left(\boldsymbol{I} - \beta \boldsymbol{A}\right)^{-1}$,
\begin{align}
     \left(\boldsymbol{I} - \beta \boldsymbol{A}\right) \boldsymbol{\gamma}_0 &= \alpha \boldsymbol{\lambda},
\end{align}
gives
\begin{align}
    \beta \underbrace{\left(\boldsymbol{I} - \boldsymbol{A}\right)}_{\boldsymbol{H}^{-1}} \boldsymbol{\gamma}_0 + \left(1-\beta\right)\boldsymbol{\gamma}_0 &= \alpha \boldsymbol{\lambda}, \\
    \implies \beta \boldsymbol{\gamma}_0 + \left(1-\beta\right) \boldsymbol{H} \boldsymbol{\gamma}_0 &= \alpha \boldsymbol{H} \boldsymbol{\lambda}.
\end{align}
In component form, this is
\begin{align}
    \beta \gamma_{0,i} + \left(1-\beta\right) \sum_{j=1}^N H_{ij} \gamma_{0,j} = \alpha \sum_{j=1}^N H_{ij} \lambda_{j}.
\end{align}
We can substitute in the decompositions of $\lambda_{j}$ and $\gamma_{0,j}$ as in eq. \eqref{e:gammadecomposed} to find,
\begin{align}
    \beta \gamma_{0,i} + \left(1-\beta\right) \sum_{j=1}^N H_{ij} \left( \bar{\gamma} + \Delta\gamma_{0,j}  \right) = \alpha \sum_{j=1}^N H_{ij} \left( \bar{\lambda} + \Delta \lambda_{0,j} \right),
\end{align}
and take the expectation of this expression conditioned on the output multiplier,
\begin{align}
\beta \expeccond{\gamma_{0,i}}{\OM_i} + \left(1 - \beta \right) \sum_{j=1}^N H_{ij} \bar{\gamma}
= \alpha \sum_{j=1}^N  H_{ij} \bar{\lambda},
\end{align}
which can be rewritten,
\begin{align}
\beta \expeccond{\gamma_{0,i}}{\OM_i} + \left(1 - \beta\right) \OM_i \bar{\gamma} = \alpha \OM_i \bar{\lambda},
\end{align}
and rearranged to
\begin{align}
\expeccond{\gamma_{0,i}}{\OM_i}= \left( \alpha \bar{\lambda} +  \left(\beta - 1\right) \bar{\gamma} \right) \OM_i / \beta. \label{e:gammaandOM}
\end{align}
From eq. \eqref{e:gammaandOM}, $\expeccond{\gamma_{0,i}}{\OM_i}$ increases with $\OM_i$ if $$ \bar{\lambda} > \frac{\left(1 - \beta\right)}{\alpha} \bar{\gamma},$$ which is the inequality \ref{p:corrOMgamma}. When the elasticity with respect to inputs is sufficiently large ($\beta \rightarrow 1$), then the term on the right hand side of this inequality shrinks to 0. Our model thus suggests that productivity improvements are expected to be a positive function of the output multiplier, given sufficiently high industrial interdependency. 


\section{Policy design in the first order}\label{S:HultenExperiment}

Whilst eq. \eqref{Hultens_Result} is useful in an accounting framework, policy-makers affect inputs to TFP-growth rather than TFP itself. Here we explore how economic growth ($\g$) responds to changes in the growth rate of endowments ($\lambda_i$), i.e.
\begin{equation}
    \pdv{\g}{\lambda_i}.
\end{equation}
Using the model developed in Section \ref{S:interdepTFP}, we  show how the relation in eq. \eqref{Hultens_Result} can be derived as a special case of this computation.

\begin{prop} 
For the IO model described in Section \ref{S:StandardIO}, the effect of a change in $\lambda_j$ in the $j$-th industry on economic growth $\g$ is
\begin{align}
    \pdv{\g}{\lambda_j} &= \alpha \sum_{k=1}^N \underbrace{H_{\beta,ik} \theta_k  }_{(1)}+  \underbrace{\gamma_{0,k} \pdv{\theta_k}{\lambda_j}}_{(2)},
    \label{e:breakhulten}
\end{align} 
where $\domar_k$ is the Domar weight of industry $k$.\label{t:breakhulten}
\end{prop}
Proposition \ref{t:breakhulten} breaks down the effect of an increase in $\lambda_j$ on $\g$ into (1), a sum of Domar weight across \textit{all} industries \textit{weighted} by a network statistic, and (2), the product of the growth rate of TFP at the steady state and the change in the Domar weights of all industries. In contrast to eq. \eqref{Hultens_Result}, Proposition \ref{t:breakhulten} captures the effect of changes in productivity growth due to the IO structure of the economy. Changes in the inputs to TFP-growth depend on \textit{all} industries and the structure of the economy through the distribution of the Domar weights, even to a first order approximation.

We can relate Proposition \ref{t:breakhulten} to Hulten's result in Corollary \ref{t:hultendynamic} given three conditions. First, that there are no productivity spillovers between industries, so the elasticity of TFP growth in the $i$-th industry with respect to that in the $j$-th is $\beta = 0$. Then, the steady state is $\gamma_0 = \alpha \boldsymbol{\lambda}$, and $\boldsymbol{H}_{\beta} = \boldsymbol{I}$. Second, that productivity $Z_i$ grows exponentially in time, and moves 1-for-1 to changes in factors of productivity growth $E_i$ so $\alpha = 1$ and $\boldsymbol{\lambda} = \boldsymbol{1}$. Third, that Domar weights are fixed, so $\pdv{\domar_k}{\lambda_j}=0$. Under these conditions, eq. \eqref{e:breakhulten} is
\begin{align}
    \pdv{\g}{\lambda_j} &=  \sum_{k=1}^N H_{\beta,jk} \theta_k = \theta_j
    \label{THELASTLASTONE},
\end{align} 
which is the same relationship between productivity changes and $\g$ in Hulten's result. The advantage is that we have considered how policy affects $\gamma_i$ through $\lambda_i$, rather than compute $\pdv{\g}{\gamma_i}$ directly.

\section{Cobb-Douglas and the importance of backward-linkages}\label{Cobb-Douglas}
To show network effects neglected by previous literature, the following section evaluates Proposition \ref{t:breakhulten} on a network of industries that complies with Hulten's Theorem in \citep{Baqaee2019}. We assume industries produce following a Cobb-Douglas production function (\textbf{A6}),

\begin{equation}
F_i(\boldsymbol{X}_i, L_i, Z_i) = Z_i \eta_i L_i^{\alpha_{iL}} \displaystyle\prod_{j=1}^N X_{ji}^{\alpha_{ij}}, \label{CobbDouglas}
\end{equation}
where the variable $\eta_i$ is a normalisation constant\footnote{The constant helps derive steady state values for inputs to the production function, as in \citep{Carvalho2019}, and is set as $\eta_i = (\alpha_{iL}^{\alpha_{iL}} \prod_{j=1}^N \alpha_{ij}^{\alpha_{ij}})^{-1}$.}, and the terms $\alpha_{ij}$ and $\alpha_{iL}$ are defined in $\rp$ such that  $\alpha_{iL} + \sum_{j=1}^N\alpha_{ij} = 1$. The profit maximisation schedule of each industries implies allocative efficiency. The adjacency matrix of this model is then
\begin{equation}
    \boldsymbol{A} = \left( \begin{matrix}
        \alpha_{11} & ... & \alpha_{1n} \\
        \vdots & \ddots & \vdots \\
        \alpha_{n1}& ... & \alpha_{nn} \\
    \end{matrix} \right),
\end{equation}
and we also have the following result. 
\begin{prop}
For the IO model described in Section \ref{S:StandardIO} assuming \textbf{A1}-\textbf{A6}, the response of economic growth to policy changes is given by
\begin{equation}
\pdv{\g}{\lambda_j} = \alpha \sum_{k=1}^N H_{\beta,jk} \theta_k.
\label{e:Cobbdougals_result}
\end{equation}
\label{p:CobbDouglas}
\end{prop}
The result is Proposition \ref{t:breakhulten} where the second term (2) is null, as shown in \cite{Baqaee2019}'s study of changes in Domar weights. This is made explicit when writing
\begin{equation}
 \pdv{\domar_k}{\lambda_j} = \sum_{q=1}^N\pdv{\domar_k}{\ln Z_q} \pdv{\ln Z_q}{\lambda_j}, \label{e:whatisdomarlambda}
\end{equation}
where by Proposition 7 in \cite{Baqaee2019} for the one factor general constant elasticity of substitution (CES) production function, $\pdv{\domar_k}{\ln Z_q} = 0$ for all $k, q = 1,...,N$. It follows that eq. \eqref{e:whatisdomarlambda} is a sum of 0 terms, and term (2) in eq. \eqref{e:breakhulten} cancels out to yield Proposition \ref{p:CobbDouglas}. 

We use the Cobb-Douglas case to clarify two claims made by our model and their difference to previous work. 

First, we presented an intuitive generalisation of Hulten's result. Considering the effect of policy changes (changes in $\lambda_i$) on economic growth ($\g$) makes this model a thought experiment for policy discussion, rather than an accounting tool. In doing so, Proposition \ref{p:CobbDouglas} like Proposition \ref{t:breakhulten} makes the case that neighbouring Domar weights, as well as the structure of the network capture by $\boldsymbol{H}_\beta$, determine changes to the rate of economic growth. Increasing the number of paths between industry $j$ and its buyers increases the off-diagonal elements of $\boldsymbol{H}_\beta$, and thus amplifies the effect of increasing $\lambda_j$. Assuming that percentage changes in $\lambda_i$ move 1-to-1 for percentage changes in $Z_i$($\alpha=1$), we can immediately  relate it to Hulten's result, e.g. $\partial \g / \partial \lambda_j = \domar_j$. 

Second, interdependent productivity growth revises Hulten's result more radically than previous work, which has so far focused on second order effects. The path-breaking work of \cite{Baqaee2019} discusses the response of final demand to second-order changes in industry-level productivity, and, like Proposition \ref{t:breakhulten}, focuses on changes in Domar weights. There are three important differences. Firstly, their model is static and considers final demand rather than economic growth. Secondly, our model argues for differences in the first order, a more direct conflict with Hulten's result. Finally, \cite{Baqaee2019} find that when the production function is Cobb-Douglas, the second order effects are null, and Hulten's Theorem holds. In contrast, Proposition \ref{p:CobbDouglas} stresses the importance of network structure to economic growth and productivity even in a highly simplified (Cobb-Douglas), first-order setting. Unlike results in \cite{Baqaee2019}, neither Domar weights nor the structure of the economy need change for the economic structure to amplify or reduce returns to industrial investment. Whilst the model in \cite{Baqaee2019} virtuously captures industrial dynamics such as Baumol's cost disease or the effect of oil price shocks, we focus on an important channel by which industrial diversification and interdependence between industries, such as in manufacturing, may accelerate growth by \textit{amplifying} industry TFP improvements \citep{Hausmann2014TheProsperity, mcnerney2021bridging}.

\begin{figure}
        \centering
        \includegraphics[width=0.7\textwidth]{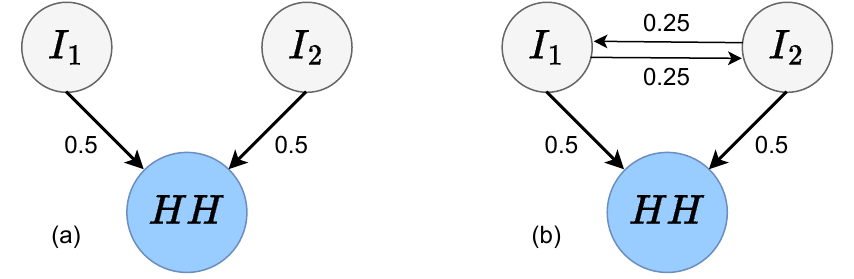}
          \caption{Hulten's result as eq. \eqref{Hultens_Result} models the effect of a productivity increase across two graphs as identical, where our model stresses the losses from of productivity gains in the interdependent economy.
          Panels (a) and (b) show two economies with industries that have identical Domar weights, given a Cobb-Douglas production function and logarithmic preferences for the representative household.
          Hulten's Theorem implies that productivity shocks are equal across both economies. The Domar Weights are equal because the row sums of the Leontief inverses are identical:
          \begin{minipage}{\linewidth}
          $$
            \boldsymbol{A}_a = \left( \begin{matrix}
            0.5 & 0 \\
            0 & 0.5 \end{matrix} \right) ,\,\,   
            \boldsymbol{A}_b = \left( \begin{matrix}
            0.25 & 0.25 \\
            0.25 & 0.25 \end{matrix} \right) 
            \implies 
            \boldsymbol{H}_a = \left( \begin{matrix}
            2 & 0 \\
            0 & 2 \end{matrix} \right), \,\, 
            \boldsymbol{H}_b =  \left( \begin{matrix}
            1.5 & 0.5 \\
            0.5 & 1.5 \end{matrix} \right),  $$
            \end{minipage}
            such that, taking preferences $\boldsymbol{\pref}=(0.5,0.5)^T$ to be symmetric across goods from industry $1$ and $2$, the Domar weights for each network are $\boldsymbol{\domar} = (1,1)^T$. However, looking at our model's gains from returns to endowments to TFP growth rates $ \lambda_i$ for $i=1,2$ we find that $\partial \g /\partial\lambda_i = 2$ in network (a), and $\partial \g/\partial \lambda_i = 1.25$ in network (b). The difference arises because productivity improvements are propagated up and down supply chain linkages.}\label{fig1}
\end{figure}

It is tempting to consider Proposition \ref{t:breakhulten} as a simple upward revision of the effect of TFP improvements. However, an insightful implication is the role that backward linkages may play in the response of economic growth. We illustrate this by considering a specific form for final demand, following \cite{Carvalho2019}. First, we note that eq. \eqref{e:Cobbdougals_result} under matrix form is:
\begin{equation}
    \nabla_{\boldsymbol{\lambda}} \g = \alpha \boldsymbol{H_{\beta}}\boldsymbol{\domar}, \label{e:CobbDouglas_result_mat}
\end{equation}
where $\nabla_{\boldsymbol{\lambda}}$ is the gradient with respect to the vector of rates of exponential growth $\boldsymbol{\lambda}$.
Now let the representative household supply one unit of labour (L=1), so $\fc = \sum_i L_i w = w$. The household consumes from $n$ goods by maximising logarithmic preferences,
\begin{equation}
    u(C_1,...,C_n) = \sum_{i=1}^N \pref_i \log (C_i/ \pref_i),
\end{equation}
such that $\sum_{i=1}^n \pref_i = 1$, and subject to the budget constraint,
\begin{equation}
	w \geq \sum_{i=1}^N C_i p_i.
\end{equation}
At optimum, $C_i = \pref_i w / p_i$ \citep{Carvalho2019}. The optimisation schedule for industries yields $X_{ij} = \frac{\alpha_{ij} Y_i p_i}{p_j}$. Assuming all goods markets clear, so that,
\begin{equation}
Y_i = C_i + \sum_{j=1}^N X_{ij},
\end{equation}
we can substitute in the optimum values, 
\begin{equation}
Y_i = \pref_i w / p_i + \sum_{j=1}^N \alpha_{ji} p_j Y_j / p_i.
\end{equation}
Multiplying by $p_i$ and dividing through by $\fc = w$ yields
\begin{equation}
Y_i p_i / \fc = \pref_i + \sum_{j=1} \alpha_{ij} p_j Y_j / \fc,
\end{equation}
where we can relabel $Y_i p_i / \fc = \domar_i$. Writing this under matrix form and rearranging, we find,
\begin{equation}
\boldsymbol{\domar} = \left( \boldsymbol{I} - \boldsymbol{A} \right) ^{-1} \boldsymbol{\pref} = \boldsymbol{H} \boldsymbol{\pref}. \label{e:domar_CB}
\end{equation}
Substituting this value into eq. \eqref{e:CobbDouglas_result_mat} we have,
\begin{equation}
     \nabla_{\boldsymbol{\lambda}} \g = \alpha \boldsymbol{H_{\beta}} \boldsymbol{H} \boldsymbol{\pref}.
\end{equation}
To make sense of this result, consider the case where $\beta \rightarrow 1$, capturing a high influence of neighbouring TFP growth rates. Then 
\begin{equation}
     \nabla_{\boldsymbol{\lambda}} \g \rightarrow \alpha \boldsymbol{H}^{2} \boldsymbol{\pref},
\end{equation}
that is, the response of economic growth to an improvement depends on the squared Leontief inverse. Component-wise, this is
\begin{equation}
    \pdv{\g}{\lambda_j} = \alpha \sum_{k=1}^N H_{jk}H_{kj} \pref_j.
\end{equation}
Recall that the coefficient $H_{jk}$ of $\boldsymbol{H} = \left(\boldsymbol{I} - \boldsymbol{A}\right)^{-1} = \boldsymbol{A} + \boldsymbol{A}^2 + \boldsymbol{A}^3 + ...$ is a sum of all possible weighted paths from the vertex (industry) $i$ to $j$. When $\beta \rightarrow 0$, the effect of a change in the endogenous growth rate $\lambda_j$ is given by the $j$-th coefficient of the Domar weights in eq. \eqref{e:domar_CB}. This depends solely on the importance of industry $j$ as a supplier through $H_{j1},...,H_{jN}$. When $\beta \rightarrow 1$, Proposition \ref{t:breakhulten} instead stresses that the response to changes in $\lambda_j$ should consider industry $j$'s role as a buyer of intermediate inputs -- its backward linkages -- through $H_{1j},...,H_{Nj}$. Figure \ref{fig1} shows how two networks with identical Domar weights across all industries can thus differ in returns to changes in the parameters $\lambda_i$ of TFP growth. The insights here are that new backward and forward linkages amplify the returns to productivity enhancements, and that existing linkages stresses the importance of all industries in contributing towards shared economic growth.

\section{Discussion}\label{S:Discussion}
This paper contributes to a wide literature assessing the impact of industrial interdependencies on economic growth. The standard IO model provides a clear framework in which to define measures of the importance of an industry to the economy's supply chain, such as the output multiplier. The work of \cite{Mcnerney2021HowGrowth} directly relates economic growth to the weighted average of industry output multipliers and weighted average productivity growth. Their work is succinct, tractable, and relates to previous literature through Hulten's result. However, \cite{Mcnerney2021HowGrowth} choose  to ignore a positive correlation between output multiplier and TFP growth rates. We see two issues. First, the finding occludes policy targeting TFP if the policy risks also affecting industry centralities. Second, it contradicts Hulten's result's assumption of the network structure being irrelevant to TFP-growth. To address these issues, we explored the question of TFP spillovers over IO linkages. In developed economies, productivity spillovers along IO linkages have been important to explaining the gains of the Computer revolution \citep{Raa2006TheInefficiency, Wolff2000EnginesRaa}. In emerging economies, ideas with empirical support like the `principal of relatedness' \citep{mcnerney2021bridging} suggest productivity gains should diffuse across exchanges. There thus appears a need to reconcile IO models with interdependent TFP growth.

To address these challenges, this paper defined interdependent TFP growth in terms of stocks drawing on the form of a Cobb-Douglas knowledge production function \citep{ErturKoch2007, pichler2018technological}. Deriving growth rates for TFP growth as a Lotka-Volterra system produced a steady state value which was embedded into the IO model. Within this framework, the correlation between productivity growth and output multiplier is expected to be a positive function of interdependencies, as captured by increasing elasticity of industry TFP growth regarding neighbouring growth levels. Our model, however, may come at the expense of Hulten's result. Rather than reducing a policy intervention to a change in industry-level TFP aggregated by Hulten's Theorem, our model considers how a policy's effects are transferred through TFP and the entire IO network. Like \cite{Baqaee2019}, a deviation from Hulten's result in eq. \eqref{Hultens_Result} came from changing Domar weights. The difference in our model is that this holds even in a first order approximation with  perfect competition. We showed how to recover Hulten's result as a special case of complete industry TFP independence. Finally, considering a specific homothetic, convex utility function for final demand provides our model with insight on the qualitative differences to Hulten's result. Specifically, we find that productivity shocks may differ in consequences for industries with identical Domar weights because productivity shocks are propagated upstream, to the suppliers of the shocked industry. We thus reinforce the importance of the entire economic environment to industry level TFP.

There are important omissions in our work. We do not seek to develop real-world frictions. This is to provide a clear channel through which TFP conditioned on existing IO linkages determines the rate of economic growth. Within our model of TFP growth, however, our assumption of constant industry-specific growth rates is precarious. Most industries display changing returns to R\&D throughout their development cycle \citep{Fink2019HowInnovation, Haraguchi2019StructuralEmployment}. Asking how industry-level productivity growth rates naturally changes over time  might provide further insight into the dynamics of structural change, as emerging economies shift through industries \citep{Lin2013LessonsTransformation, mcnerney2021bridging}. In providing intuition of our result that contrasts Hulten's result \citep{Hulten1978GrowthInputs}, we limit ourselves to exploring the case of Cobb-Douglas production functions to abstract from changes in Domar weights. However, non-linear responses of Domar weights are discussed in depth by \cite{Baqaee2019} within a static model, and their integration is an exciting avenue for future research. Our conclusion that the Domar weight of industries is a poor aid to development policy seeking to optimise investments for productivity improvements is not novel. However, our model provides a potential framework to discuss the importance of backwards and forward linkages, whilst remaining compatible with previous IO designs. The takeaway that backward as well as forward linkages may be crucial to optimising the effect of increasing TFP growth rates stresses a need to develop a well-balanced domestic economy for long-term prosperity, an insight that may help explain the struggle of emerging markets specialised in highly independent export sectors, and the success of manufacturing-led, interdependent, export economies. 

\appendix
\section{Appendix}

We first list together the assumptions of the IO, TFP and joint models in Appendix \ref{Assumptins_Appendix}. Then, we give proofs based on the three sections of this paper: those derived from the model in \citep{Mcnerney2021HowGrowth} in \ref{Standard_IO_Appendix}, those derived for the model of interdependent TFP growth rates in \ref{TFP_Appendix}, and finally discuss the invertability of $H_{\beta}$ in \ref{Appendix_LV_Invertibility}, exploiting results on Lotka-Volterra systems.

\subsection{Assumptions}
\label{Assumptins_Appendix}
For clarity, we restate the assumptions made to construct the IO model.

\begin{align}
&\textbf{(A1)} \ 
&\pdv{F_i}{k} > 0, \text{ }\pdvtwice{F_i}{k} &< 0, \text{ }\forall {k} \in \{\boldsymbol{X}_i L_i\} \tag{diminishing marginal returns}\\
&\textbf{(A2)} \ 
&F_i(bX_{i1},...,b L_i,Z_i) &= b F_i(X_{i1},...,L_i,Z_i)\tag{constant returns to scale} \\
&\textbf{(A3)} \ 
& \underbrace{S_i + \sum_{j=1}^n m_{ji}}_{\text{flows into industry $i$}} &= \underbrace{L_i w + \sum_{j=1}^n m_{ij}}_{\text{flows out industry $i$}} \tag{perfect competition}\\
&\textbf{(A4)} \ 
& \underbrace{\sum_{i=1}^n L_i w}_{\text{flows into households}} &= \underbrace{\sum_{i=1}^n C_i p_i}_{\text{flows out}} \tag{market clearing} \\
&\textbf{(A5)} \ 
&\dot{Z}_i  &\equiv  ({E_i}e^{\lambda_i t})^{\alpha}\exoindus_i  \prod_{j=1}^N Z_{j} ^{\beta a_{ij}} \tag{Cobb-Douglas TFP}\\
&\textbf{(A6)} \ 
&F_i(\boldsymbol{X}_i, L_i, Z_i) &= Z_i \eta_i L_i^{\alpha_{iL}} \displaystyle\prod_{j=1}^N X_{ji}^{\alpha_{ij}} \tag{Cobb-Douglas production}
\end{align}

\subsection{The standard IO model}\label{Standard_IO_Appendix}
The following Lemma is taken from \cite{Mcnerney2021HowGrowth}:
\begin{lemma}
For the IO model of Section \ref{S:StandardIO} with assumptions \textbf{A1}-\textbf{A4}, TFP growth rates are $$\gamma_i = \hat{Y}_i - \displaystyle\sum_{j=1}^N a_{ji} \hat{X}_{ij} - \tilde{\ell}_{i} \hat{L}_i.$$ \label{L:prodimprov}
\end{lemma}
\begin{proof} 
The result follows from taking the time derivative of the logarithm of the production function and the definition of elasticity. Given \textbf{A1}-\textbf{A4}, allocative efficiency holds as discussed in Section \ref{S:StandardIO}. Taking the total time-derivative of the logarithm of $\indprod$ and applying the chain rule,
\begin{align}
\dv{\ln \indprod_i}{t} 
&= 
    \displaystyle\sum_{j=1}^N 
    \pdv{\ln(\indprod_i)}{\ln(X_{ji})}
    \pdv{\ln(X_{ji})}{t} +
    \pdv{\ln(\indprod_i)}{\ln(L_{i})}
    \pdv{\ln(L_{i})}{t} +
    \pdv{\ln \indprod_i}{\ln(Z_i)}\pdv{\ln(Z_i)}{t}.
\label{e:lnfdt}
\end{align}
We compute the last term on the right as, 
\begin{align}
\pdv{\ln(\indprod_i)}{\ln(Z_i)}\pdv{\ln(Z_i)}{t} &= \pdv{\indprod_i}{Z_i} \frac{Z_i}{\indprod_i} \frac{1}{Z_i} \dot{Z}_i \\
&= \frac{G_i Z_i}{\indprod_i}\hat{Z}_i = \hat{Z}_i = \gamma_i.
\end{align}
Then eq. \eqref{e:lnfdt} is
\begin{align} \frac{\dot{\indprod_i}}{\indprod_i} &= 
    \displaystyle\sum_{j=1}^N 
    \frac{\partial \ln(\indprod_i)}{\partial \ln(X_{ji})}
    \frac{\dot{X_{ji}}}{X_{ji}} +
    \frac{\partial \ln(\indprod_i)}{\partial \ln(L_{i})}
    \frac{\dot{L_{i}}}{L_i} +
    \gamma_i.
\end{align}
Using
\begin{align}
\pdv{\ln(\indprod_i)}{\ln(X_{ji})} =  \text{Elasticity}\left(\indprod_i, X_{ji}\right),
\end{align}
and by assumption of allocative efficiency, Elasticity$(\indprod_i, X_{ji}) = a_{ij}$. Hence
\begin{flalign}
 \frac{\dot{\indprod_i}}{\indprod_i} &=
    \displaystyle\sum_{j=1}^N 
    \s
    \frac{\dot{X_{ji}}}{X_{ji}} +
    \ls
    \frac{\dot{L_{i}}}{L_i} +
    \gamma_i,
 \end{flalign}
which rearranging and denoting growth rates with hats gives
\begin{equation*}
\gamma_i = \hat{\indprod}_i - \displaystyle\sum_{j=1}^N a_{ji} \hat{X}_{ij} - \tilde{\ell}_{i} \hat{L}_i. \qedhere
\end{equation*}
\end{proof}

\begin{proof}[Proof of Theorem \ref{t:pricechanges}]
We start from \textbf{A3}, a conservation of mass argument for industries,
\begin{align}
    \sum_{j=1}^N m_{ij} + L_i w = \sum_{j=1}^N X_{ji} p_j + L_i w
\end{align}
Now consider the total time derivative of the logarithm of this expression:
\begin{align}
    \dv{\ln(\indprod_i p_i)}{t} =& \dv{ \ln(\sum_{j=1}^N X_{ji} p_{j} + L_i w)}{t}.
\end{align}
Accounting for price changes, we can apply the chain rule,
\begin{multline}
\hat{p_i} + \hat{\indprod_i} =
    \displaystyle\sum_{j=1}^N
        \frac{\partial \ln(\sum_{j=1}^N X_{ji} p_{j} + L_i w)}{\partial \ln(X_{ji})}\frac{\partial \ln(X_{ji})}{\partial t} +
        \frac{\partial \ln(\sum_{j=1}^N X_{ji} p_{j} + L_i w)}{\partial \ln(L_i)}\frac{\partial \ln(L_i)}{\partial t} \\
         +
        \frac{\partial \ln(\sum_{j=1}^N X_{ji} p_{j} + L_i w)}{\partial \ln(p_j)}\frac{\partial \ln(p_j)}{\partial t} 
         +
        \frac{\partial \ln(\sum_{j=1}^N X_{ji} p_{j} + L_i w)}{\partial \ln(w)}\frac{\partial \ln(w)}{\partial t}.
\end{multline}
We compute the partial derivatives directly,
\begin{multline}
\hat{p_i} + \hat{\indprod_i} = \displaystyle\sum_{j=1}^N
    \frac{X_{ji}p_j}{\sum_{j=1}^N X_{ji} p_{j} + L_i w} \hat{X_{ji}} +
    \frac{X_{ji}p_j}{\sum_{j=1}^N X_{ji} p_{j} + L_i w} \hat{p_{j}} \\
    + \frac{L_{i}w}{\sum_{j=1}^N X_{ji} p_{j} + L_i w} \hat{L_{i}} +
    \frac{L_{i}w}{\sum_{j=1}^N X_{ji} p_{j} + L_i w} \hat{w}.
\end{multline}
Again using the conservation of mass argument $\sum_{j=1}^N X_{ji} p_{j} + L_i w = m_i$, and $X_{ji} p_j = m_{ij}$,
\begin{align}
    \hat{p_i} + \hat{\indprod_i} &= 
    \displaystyle\sum_{j=1}^N
    \frac{m_{ij}}{m_i}(\hat{X_{ji}} + \hat{p_{j}}) +
    \frac{L_{i}w}{m_i}(\hat{L_{i}} + \hat{w}),
\end{align}
where $\frac{m_{ij}}{m_i}$ is the coefficient $a_{ij}$ of the IO network's adjacency matrix, and $\frac{L_i w}{m_i}$ is $\ls_i$. Hence this simplifies to
\begin{align}
    \hat{p_i} + \hat{\indprod_i} = \sum_{j=1}^N
    \s(\hat{X_{ji}} + \hat{p_{j}}) +
    {\ls}(\hat{L_{i}} + \hat{w}),
\end{align}
which, once rearranged, is:
\begin{align}
    \underbrace{\hat{\hat{p}_i} - \hat{w}}_{= r_i} &= -\underbrace{\left( \hat{\indprod}_i - \sum_{j=1}^N \s \hat{X_{ji}} - \hat{L}_i \ls \right)}_{=\gamma_i \text{from Theorem \ref{t:pricechanges}}} + \sum_{j=1}^N\underbrace{(\hat{p}_j - \hat{w})}_{=r_j}\s.
\end{align}
Taking
\begin{align}
    \hat{r}_i &= -\gamma_i + \displaystyle\sum_{j=1}^N \hat{r}_j a_{ij},
\end{align}
and rearranging to write in matrix form gives 
\begin{equation}
\boldsymbol{\hat{r}} = - (\boldsymbol{I} - \boldsymbol{A})^{-1} \boldsymbol{\gamma} = - \boldsymbol{H} \boldsymbol{\gamma} \qedhere \label{price level}
\end{equation}
\end{proof}

\begin{proof}[Proof of Corollary 1]
The following explicates the proof in the appendix of \cite{Mcnerney2021HowGrowth}. We use Theorem \ref{t:OMg} to express $\g$ as
\begin{align}
    \g = \tilde{\gamma} \bar{\mathcal{L}},
\end{align}
where we expand in the definition of $\tilde{\gamma} \equiv \boldsymbol{\gamma} \cdot \boldsymbol{m}/\mathcal{M}$,
\begin{align}
    \g = \boldsymbol{\gamma} \cdot \frac{\boldsymbol{m}}{\mathcal{M}} \bar{\mathcal{L}},
\end{align}
to apply Lemma \ref{l:OMavrg} (Appendix \ref{Standard_IO_Appendix}) for $\bar{\mathcal{L}} = \frac{\mathcal{M}}{{S}}$, so
\begin{align}
    \g &= \boldsymbol{\gamma} \cdot \frac{\boldsymbol{m}}{\mathcal{M}} \frac{\mathcal{M}}{{S}} = \boldsymbol{\gamma} \cdot \boldsymbol{m}\frac{1}{\fc}.
\end{align}
where Hulten's result is obtained by deriving with respect to $\gamma_i$, 
\begin{align}
\pdv{\g}{\gamma_i} & = m_i/S = \domar_i. \qedhere
\end{align}
\end{proof}

\begin{lemma}
For the IO model in Section \ref{S:StandardIO} with assumptions \textbf{A1}-\textbf{A4}, economic growth is $\g \equiv \hat{\boldsymbol{C}} \cdot \boldsymbol{\tilde{c}} - \hat{L}$. \label{l:g=CC-r}
\end{lemma}
\begin{proof}[Proof]
We show the derivation of economic growth, as $\g = \hat{\fc} - \boldsymbol{\hat{p}} \cdot \boldsymbol{\tilde{c}}$, which is the growth of total household spending minus the average growth rate of prices weighted by spending shares. We take the time-differentail of  eq. \eqref{e:COM_households} and apply the chain rule, 
\begin{align}
\dot{\fc} = \dv{}{t} \sum_{i=1}^N \frac{C_i p_i}{Lw}= \sum_{i=1}^N \frac{\dot{C}_i p_i}{Lw} - \frac{C_i p_i \dot{L}}{L^2} + \frac{\dot{p}_i C_i}{Lw}, 
\end{align}
and introduce $\frac{C_i}{C_i}$ and $\frac{p_i}{p_i}$,
\begin{align}
\dot{\fc} = \sum_{i=1}^N \hat{C_i} \frac{C_i p_i}{Lw} - \frac{C_i p_i}{Lw} \hat{L} + \hat{p_i}\frac{p_i C_i}{Lw} = \hat{\boldsymbol{C}} \cdot \boldsymbol{\tilde{c}} - \hat{L} + \boldsymbol{\tilde{c}} \cdot \boldsymbol{\hat{p}}.
\end{align}
Then subtracting  $\boldsymbol{\tilde{c}} \cdot \boldsymbol{\hat{p}}$ from both sides we have,
\begin{align}
   \dot{\fc} - \boldsymbol{\tilde{c}} \cdot \boldsymbol{\hat{p}} = \hat{\boldsymbol{C}} \cdot \boldsymbol{\tilde{c}} - \hat{L} = \g.
\end{align}
\end{proof}

\begin{theorem}
For the IO model in Section \ref{S:StandardIO} with assumptions \textbf{A1}-\textbf{A4}, economic growth is $\g = -\hat{r}$. \label{t:growthasdeflation}
\end{theorem}

\begin{proof}
    The proof starts from the conservation of mass argument for the household vertex, eq. \ref{e:COM_households},
\begin{align}
\underbrace{Lw}_{\text{household income}} & = \underbrace{\textbf{C} \cdot {\textbf{p}}}_{\text{household consumption}} = \fc,
\end{align}
 where dividing through by worker-headcount L, and deriving with respect to time this is
\begin{align}
\dot{w} &= \dot{\textbf{C}} \cdot \textbf{p}/L - {\textbf{C}} \cdot \frac{\dot{L}}{L^2} \textbf{p} + \textbf{C} \cdot \dot{\textbf{p}}/L.
\end{align}
This can be expanded as: 
\begin{align}
\dot{w} &= 
\displaystyle\sum_{i=1}^{N} 
    \frac{\dot{C} _i p_i}{L} 
    - \frac{C_i \dot{L} p_i}{L^2}
    + \frac{\dot{p}_i C_i}{L} = \frac{1}{L} \left(
    \sum_{i=1}^{N} \frac{C_i}{C_i} \dot{C}_i p_i
    - \sum_{i=1}^{N} C_i p_i \hat{L}
    + \sum_{i=1}^{N} \dot{p}_i C_i \right),
\end{align}
or, multiplying the first and last terms by $1 = \frac{\sum_{k=1}^{N} C_k p_k} {\sum_{k=1}^{N} C_k p_k} = \frac{\fc}{\fc}$, 
\begin{align}
\dot{w}=\frac{1}{L} \left(
    \frac{S}{S} 
    \sum_{i=1}^{N} \frac{C_i}{C_i} \dot{C}_i p_i
    - S\hat{L}
    + \sum_{i=1}^{N} \dot{p}_i C_i \frac{\fc} {\fc} \right).
\end{align}
Factoring out $S$ and multiplying the last term by $p_i/p_i = 1$,
\begin{align}
\dot{w}= \frac{S}{L} \left(
    \sum_{i=1}^{N} \frac{\dot{C}_i}{C_i}  \frac{C_i}{S}
    - \hat{L} +  \sum_{i=1}^{N} \frac{\dot{p}_i }{p_i} \frac{p_iC_i}{S}\right).
 \end{align}
Using the accounting identity again to see that $\frac{S}{L} = w$ and the definitions of the growth rate of labour $\hat{L}$, prices $\hat{p}$ and consumption $\hat{C}$, we can write this as 
 \begin{align}
\dot{w} &= w\left(\boldsymbol{\tilde{c}} \cdot \boldsymbol{\hat{p}}  - \hat{L} + \hat{\boldsymbol{C}} \cdot \boldsymbol{\tilde{c}} \right),
\end{align}
which rearranges to:
\begin{align}
 - \boldsymbol{\hat{p}} \cdot  \boldsymbol{\tilde{c}} + \frac{\dot{w}}{w} = - \hat{L} + \hat{\boldsymbol{C}} \cdot \boldsymbol{\tilde{c}}.
\end{align}
From Lemma \ref{l:g=CC-r}, $\g = - \hat{L} + \hat{\boldsymbol{C}} \cdot \boldsymbol{\tilde{c}}$, and by definition $\hat{r} = \boldsymbol{\hat{p}} \cdot  \boldsymbol{\tilde{c}} - \frac{\dot{w}}{w}$, so 
\begin{align}
 \g &= -\hat{r}. \qedhere
\end{align}
\end{proof}

\begin{lemma}
For the IO model in Section \ref{S:StandardIO} with assumptions \textbf{A1}-\textbf{A4}, the weighted average multiplier $\bar{\mathcal{L}}$ can be expressed as $\bar{\mathcal{L}} = \frac{\mathcal{M}}{\fc}$ \label{l:OMavrg}.
\end{lemma}

\begin{proof}[Proof of Lemma \ref{l:OMavrg} and Theorem \ref{t:OMg}]
Consider the conservation of mass for industries in eq. \eqref{e:COM_industries},
\begin{align}
 \sum_{j=1}^n m_{ji} + C_i  p_i &= \displaystyle\sum_{j=1}^N m_{ij} + L_i w,
\end{align} 
where the left-hand side is industry revenue $m_i$. Then normalising the first term on the right hand side to obtain the coefficient of the adjacency matrix
\begin{align}
m_i &=  \sum_{j=1}^N \frac{X_{ji} p_i}{p_j X_j} (X_j p_j) + S_i = \sum_{j=1}^N a_{ji} (X_j p_j) + S_i,
\end{align} 
which in matrix form is
\begin{align}
\boldsymbol{m} &= \boldsymbol{A}^{T} \boldsymbol{m} + \boldsymbol{S}.
\end{align} 
This rearranges to
\begin{align}
\boldsymbol{m} &= (\boldsymbol{I} - \boldsymbol{A}^{T})^{-1} \boldsymbol{S},
\end{align} 
which if we take the transpose on both sides is
\begin{align}
\boldsymbol{m}^{T} &= \boldsymbol{S}^{T}(\boldsymbol{I} - \boldsymbol{A})^{-1}= \boldsymbol{S}^{T} \boldsymbol{H}.
\end{align} 
We look to insert weighted productivity improvements $\tilde{\gamma} \equiv \boldsymbol{\gamma} \cdot \boldsymbol{m}/\mathcal{M}$ into this expression. Note that dividing $\boldsymbol{S}$ by $\fc$ and using eq. \eqref{e:COM_households} gives components $S_i / \fc = C_i p_i / \sum_{j=1}^N C_j p_j = C_i p_i/(Lw) = \tilde{c}_i$, so that
\begin{align}
\boldsymbol{m}^{T}\frac{1}{\fc} &= \boldsymbol{\tilde{c}}^{T} \cdot \boldsymbol{H},
\end{align} 
by taking the dot product with $\boldsymbol{\gamma}$ and some manipulation, we obtain
\begin{align}
\frac{1}{\fc}  \boldsymbol{m}^{T} \cdot \boldsymbol{\gamma} &= \boldsymbol{\tilde{c}}^{T} \cdot \boldsymbol{H} \cdot \boldsymbol{\gamma}, \\
\frac{1}{\fc} \boldsymbol{m}^{T} \cdot \boldsymbol{\gamma} \frac{\mathcal{M}}{\mathcal{M}}&= \boldsymbol{\tilde{c}}^{T} \cdot \boldsymbol{H} \cdot \boldsymbol{\gamma}, \\
\frac{\mathcal{M}}{\fc} \tilde{\gamma} &= \boldsymbol{\tilde{c}}^{T} \cdot \boldsymbol{H} \cdot \boldsymbol{\gamma}. \label{e:O/Ygamma}
\end{align}
To find $\g$ in this expression, we now prove Lemma \ref{l:OMavrg}. We can express the right-hand side of eq. \eqref{e:O/Ygamma} in terms of $\bar{\mathcal{L}}$ by noting that eq. \eqref{e:O/Ygamma} holds for any $\boldsymbol{\gamma}$. Thus we can choose $\boldsymbol{\gamma} = \tilde{\gamma} \boldsymbol{1}$, such that
\begin{align}
\boldsymbol{\tilde{c}}^{T} \cdot \boldsymbol{H} \cdot \boldsymbol{1}\tilde{\gamma} &=\frac{\mathcal{M}}{\fc} \tilde{\gamma}
\end{align}
which dividing by $\tilde{\gamma}$ is
\begin{align}
\boldsymbol{\tilde{c}}^{T} \cdot \boldsymbol{H} \cdot \boldsymbol{1} = \bar{\mathcal{L}} = \frac{\mathcal{M}}{\fc},
\end{align}
proving Lemma \ref{l:OMavrg}. We substitute $\frac{\mathcal{M}}{\fc}$ for $\bar{\mathcal{L}}$ in eq. \eqref{e:O/Ygamma} so that
\begin{equation}
\bar{\mathcal{L}} \tilde{\gamma} = \boldsymbol{\tilde{c}}^{T} \cdot \boldsymbol{H} \cdot \boldsymbol{\gamma}.
\end{equation}
Finally, we can write this in terms of $\g$ on the right hand side by applying Theorem \ref{t:growthasdeflation} ($\g = - \hat{r}$), the definition of real inflation ($\hat{r} = \boldsymbol{\tilde{c}} \cdot \boldsymbol{\hat{r}}$), and Theorem \ref{t:pricechanges} ($\hat{r}= -\boldsymbol{H} \cdot \boldsymbol{\gamma}$), which gives
\begin{equation}
    \bar{\mathcal{L}} \tilde{\gamma}=\g. \qedhere
\end{equation}
\end{proof}

\subsection{Proofs on the model of interdependent productivity growth}\label{TFP_Appendix}
\begin{proof}[Proof of Proposition \ref{t:breakhulten}]
First, note that the response of $\gamma_{0,i}$ for $i = 1,...,N$ to industrial policy targetting $\ind_j$ is 
\begin{align}
    \pdv{\gamma_{0,i}}{\lambda_j} = \alpha \pdv{}{\lambda_j}  \sum_{k=1}^N H_{\beta,ik} \lambda_k = \alpha H_{\beta,ij}.\label{e:gammalambda}
\end{align}
In words, a change in endogenous factors $\lambda_j$ in the $j$-th industry affects all industries through direct links and indirect paths. In fact, when $\beta$ tends to $1$, this change is proportional to the sum of all walks between industry-vertex $i$ and $j$.

From Lemma \ref{t:OMg},
\begin{equation}
    \g = \tilde{\gamma} \bar{\mathcal{L}},
\end{equation}
in which we expand the definition of $\tilde{\gamma} = \boldsymbol{\gamma} \cdot \boldsymbol{m}/\mathcal{M}$, for $\boldsymbol{m}$ industry revenues and $\mathcal{M}$ the sum of all industry revenues. Then at the steady state:
\begin{equation}
    \g = \sum_{k=1}^N \gamma_{0,k} m_k \frac{\bar{\mathcal{L}}}{\mathcal{M}}.
\end{equation}
Taking the partial derivative with respect to $\lambda_j$,
\begin{align}
    \pdv{\g}{\lambda_j} &= \sum_{k=1}^N \pdv{}{\lambda_j} \gamma_{0,k} m_k \frac{\bar{\mathcal{L}}}{\mathcal{M}} = \sum_{k=1}^N m_k \frac{\bar{\mathcal{L}}}{\mathcal{M}} \pdv{\gamma_{0,k}}{\lambda_j} +  \gamma_{0,k} \pdv{}{\lambda_j} m_k \frac{\bar{\mathcal{L}}}{\mathcal{M}}, 
\end{align}
applying the result in eq. \eqref{e:gammalambda},
\begin{align}
    \pdv{\g}{\lambda_j} &= \alpha \sum_{k=1}^N H_{\beta,ij} m_k \frac{\bar{\mathcal{L}}}{\mathcal{M}}  +  \gamma_{0,k} \pdv{}{\lambda_j} m_k \frac{\bar{\mathcal{L}}}{\mathcal{M}}, 
\end{align} 
and using Lemma \ref{l:OMavrg} ($\bar{\mathcal{L}} = \frac{\mathcal{M}}{\fc}$),
\begin{align}
    \pdv{\g}{\lambda_j} &= \alpha \sum_{k=1}^N H_{\beta,kj} \frac{m_k}{\fc}  +  \gamma_{0,k} \pdv{}{\lambda_j} \frac{m_k}{\fc}.
\end{align} 
Finally, substituting $\frac{m_k}{gdp}$ by the definition of $\ind_k$'s Domar weight,
\begin{align}
    \pdv{\g}{\lambda_j} &= \alpha \sum_{k=1}^N \underbrace{H_{\beta,ij} \theta_k  }_{(1)}+  \underbrace{\gamma_{0,k} \pdv{\theta_k}{\lambda_j}}_{(2)},
\end{align} 
gives the required result.
\end{proof}

\subsection{Proof of the invertability of the community matrix for the Lotka-Volterra system}\label{Appendix_LV_Invertibility}
\begin{definition}
    A matrix $\boldsymbol{O}$ ``is called \defn{stable} if all its eigenvalues have negative real part." \citep{Allesina2022} 
    \end{definition}
    \begin{definition}
    Let $\boldsymbol{D}(\boldsymbol{y})$ be the diagonal matrix with $\boldsymbol{y}$ as the diagonal, $\boldsymbol{y}$ an $n$-dimensional vector with strictly positive coefficients. A matrix $\boldsymbol{O}$ is \defn{diagonally stable} if $\boldsymbol{D}(\boldsymbol{y}) \boldsymbol{O}$ is stable. \citep{Allesina2022}
    \end{definition}
    \begin{theorem}
    If a matrix $\boldsymbol{O}$ is ($i$) stable with ($ii$) non-negative off-diagonal elements, and ($iii$) negative diagonal elements, then $\boldsymbol{O}$ is diagonally stable \citep{DattaStability}.
    \label{t:diagonalstabilty}
    \end{theorem}
    To meet the conditions of Theorem \ref{t:diagonalstabilty}, we will need the following lemma. 
    
\begin{lemma}
Let $\leigen(\boldsymbol{O})$ be the largest of $n$ eigenvalues of a matrix $\boldsymbol{O}$ in $\mathbb{R}^{n \times n}$. Then for $k_2$ in $\rp$ and $k_1$ in $\mathbb{R}$, $\leigen(k_2 \boldsymbol{O} - k_1 \boldsymbol{I}) = k_2 \leigen(\boldsymbol{O}) - k_1$. \label{L:spectralraiduslinear}
\end{lemma}
\begin{proof} 
Starting from the eigenvalues of $\boldsymbol{O}$, we want to show that the largest eigenvalue of $(k_2 \boldsymbol{O} - k_1 \boldsymbol{I})$ is $(k_2 \leigen(\boldsymbol{O}) - k_1)$. 
Let $\boldsymbol{\omega}_i$ be any of the $n$ eigenvectors of $\boldsymbol{O}$ with corresponding eigenvalue\footnote{notation $\kappa$ is chosen for consistency, but apologies are due to the reader for the similarity to the scalars $k_1$, $k_2$.} $\kappa_i$, for $i$ in $\{1,...,n\}$. Then:
\begin{align}
\boldsymbol{O}\boldsymbol{\omega}_i = \kappa_i \boldsymbol{\omega}_i,
\end{align}
multiplying by $k_2$ and subtracting by  $k_1 \boldsymbol{\omega}_i$:
\begin{align}
(k_2\boldsymbol{O} - k_1 \boldsymbol{I})\boldsymbol{\omega}_i = (\kappa_i - k_1) \boldsymbol{\omega}_i,
\end{align}
so $k_2 \kappa_i - k_1$ is an eigenvalue of $\boldsymbol{O}$. Since, $\{\boldsymbol{\omega}_1, ... ,  \boldsymbol{\omega}_n\}$ are linearly independent, so are the eigenvectors in the set $\{(\kappa_1 - k_1) \boldsymbol{\omega}_1, ...,(\kappa_n - k_1) \boldsymbol{\omega}_n\}$. This set then spans $(\boldsymbol{O} - k_1 \boldsymbol{I})$ since this matrix is $n$-dimensional. Hence, $\{ k_2 \kappa_1 - k_1, ... ,k_2 \leigen(\boldsymbol{O}) - k_1 \}$ is the full set of eigenvalues of  $(\boldsymbol{O} - k_1 \boldsymbol{I})$. Ordering $\kappa_1 < ... < \kappa_{n-1}< \leigen(\boldsymbol{O})$ implies that $k_2 \kappa_1 - k_1 < ... < k_2 \leigen(\boldsymbol{O}) - k_1$ since $k_2 > 0$. Hence $k_2 \leigen(\boldsymbol{O}) - k_1$ is the largest eigenvalue of $k_2 \boldsymbol{O} - k_1 \boldsymbol{I}$. 
\end{proof}

If the three conditions of Theorem \ref{t:diagonalstabilty} are met by the Lotka-Volterra system of eq. \eqref{e:gammaLV}, we can infer the stability of $\boldsymbol{\gamma}_0$ by the diagonal stability of $(\beta \boldsymbol{A} - \boldsymbol{I})$. We check each in turn. The condition $(ii)$ is satisfied by the fact that $\boldsymbol{A}$ is a row-stochastic matrix, so has positive off-diagonal elements. Condition ($ii$) is satisfied by the fact that $\beta < 1$, so $\beta a_{ii} - 1 < 0$. For condition ($i$), we use Lemma \ref{L:spectralraiduslinear}. Because $\boldsymbol{A}$ is substochastic, its spectral radius is less than $1$ \citep{Carvalho2019}. Then the largest eigenvalue of $(\beta\boldsymbol{A} - \boldsymbol{I})$ is less that $\beta - 1 < 0$. Hence, $(\beta\boldsymbol{A} - \boldsymbol{I})$ is indeed stable. Then, applying Theorem \ref{t:diagonalstabilty}, $(\beta\boldsymbol{A} - \boldsymbol{I})$ must be diagonally stable. It follows that perturbations to the steady state $\boldsymbol{\gamma}_0$ is stable, that is, all perturbations from $\boldsymbol{\gamma}_0$ decay in time to a first order approximation.

\backmatter

\bmhead{Acknowledgments}

Acknowledgments are not compulsory. Where included they should be brief. Grant or contribution numbers may be acknowledged.

Please refer to Journal-level guidance for any specific requirements.

\section*{Declarations}

\subsection*{Funding}
No funds, grants, or other support was received.

\subsection*{Competing Interests}
The authors declare no competing interests.

\subsection*{Data Availability}
There is no data or code associated with this study.

\bibliography{references}

\end{document}